\documentclass[11pt,english]{article}

\usepackage{amsthm}
\usepackage{graphicx} 
\usepackage{array} 

\usepackage{amsmath, amssymb, amsfonts, verbatim}
\usepackage{hyphenat,epsfig,subcaption,multirow}
\usepackage{nicefrac}
\usepackage{paralist}

\usepackage{mathtools, etoolbox}
\usepackage{thmtools}
\usepackage{thm-restate}
\usepackage{xparse}

\usepackage[font=small,labelfont=bf]{caption}

\usepackage[usenames,dvipsnames]{xcolor}

\usepackage{babel}
\usepackage{csquotes}

\usepackage[style=alphabetic,backref=true,maxnames=99,maxalphanames=99,
            url=false,isbn=false,doi=true]{biblatex}

\AtEveryBibitem{\clearname{editor}}
\AtEveryBibitem{\clearlist{publisher}}
\AtEveryBibitem{\clearfield{pages}}
\DeclareFontFamily{U}{mathx}{\hyphenchar\font45}
\DeclareFontShape{U}{mathx}{m}{n}{
      <5> <6> <7> <8> <9> <10>
      <10.95> <12> <14.4> <17.28> <20.74> <24.88>
      mathx10
      }{}
\DeclareSymbolFont{mathx}{U}{mathx}{m}{n}
\DeclareMathSymbol{\bigtimes}{1}{mathx}{"91}

\usepackage{tcolorbox}
\tcbuselibrary{skins,breakable}
\tcbset{enhanced jigsaw}

\usepackage[normalem]{ulem}
\usepackage[compact]{titlesec}

\definecolor{DarkRed}{rgb}{0.5,0.1,0.1}
\definecolor{DarkBlue}{rgb}{0.1,0.1,0.5}

\usepackage{nameref}
\definecolor{ForestGreen}{rgb}{0.1333,0.5451,0.1333}
\definecolor{Red}{rgb}{0.9,0,0}
\usepackage[linktocpage=true,
	colorlinks,
	linkcolor=DarkRed,citecolor=ForestGreen]
	{hyperref}
\usepackage[noabbrev,nameinlink]{cleveref}
\crefname{property}{property}{Property}
\creflabelformat{property}{(#1)#2#3}
\crefname{equation}{eq}{Eq}
\creflabelformat{equation}{(#1)#2#3}

\usepackage{bm}
\usepackage{url}
\usepackage{xspace}
\usepackage[mathscr]{euscript}

\usepackage{tikz}
\usetikzlibrary{arrows}
\usetikzlibrary{arrows.meta}
\usetikzlibrary{shapes}
\usetikzlibrary{backgrounds}
\usetikzlibrary{positioning}
\usetikzlibrary{decorations.markings}
\usetikzlibrary{patterns}
\usetikzlibrary{calc}
\usetikzlibrary{fit}
\tikzset{vertex/.style={circle, black, fill=Yellow, line width=1pt, draw, minimum width=8pt, minimum height=8pt, inner sep=0pt}}

\usepackage[framemethod=TikZ]{mdframed}

\usepackage[margin=1in]{geometry}

\usepackage{soul}

\newtheorem{theorem}{Theorem}
\newtheorem{lemma}{Lemma}[section]
\newtheorem{proposition}[lemma]{Proposition}
\newtheorem{corollary}[lemma]{Corollary}

\newtheorem{problem}{Problem}

\newtheorem*{claim*}{Claim}
\newtheorem*{proposition*}{Proposition}
\newtheorem*{lemma*}{Lemma}
\newtheorem*{problem*}{Problem}

\crefname{lemma}{Lemma}{Lemmas}
\crefname{claim}{Claim}{Claims}
\crefname{enumi}{step}{steps}

\newenvironment{result}{\begin{mdframed}[backgroundcolor=lightgray!40,topline=false,rightline=false,leftline=false,bottomline=false,innertopmargin=2pt]}{\end{mdframed}}

\newtheorem{remark}[lemma]{Remark}

\theoremstyle{definition}

\newtheorem{mdproblem}{Problem}

\newtheorem*{mdproblem*}{Problem}
\newenvironment{Problem*}{\begin{mdframed}[hidealllines=false,innerleftmargin=10pt,backgroundcolor=gray!10,innertopmargin=5pt,innerbottommargin=5pt,roundcorner=10pt]\begin{mdproblem*}}{\end{mdproblem*}\end{mdframed}}
\newtheorem{mddefinition}[lemma]{Definition}

\newenvironment{abox}{\begin{tcolorbox}[
		enlarge top by=5pt,
		enlarge bottom by=5pt,
		breakable,
		frame hidden,
		overlay broken = {
			\draw[line width=1pt, black]
			(frame.north west) rectangle (frame.south east);},
		overlay = {
			\draw[line width=1pt, black]
			(frame.north west) rectangle (frame.south east);},
		boxsep=0pt,
		left=4pt,
		right=4pt,
		top=10pt,
		arc=0pt,
		boxrule=1pt,toprule=1pt,
		colback=white
	]}
{\end{tcolorbox}}

\newtheorem{mdalg}{Algorithm}
\newenvironment{Algorithm}{\begin{abox}\begin{mdalg}}{\end{mdalg}\end{abox}}

\renewcommand{\qed}{\nobreak \ifvmode \relax \else
      \ifdim\lastskip<1.5em \hskip-\lastskip
      \hskip1.5em plus0em minus0.5em \fi \nobreak
      \vrule height0.75em width0.5em depth0.25em\fi}

\renewcommand{\leq}{\leqslant}
\renewcommand{\geq}{\geqslant}
\renewcommand{\le}{\leqslant}
\renewcommand{\ge}{\geqslant}

\title{$(\Delta + 1)$ Vertex Coloring in $O(n)$ Communication}
\author{
  Maxime Flin\thanks{(\texttt{maximef@ru.is})
    Reykjavik University.
    Supported by the Icelandic Research Fund (grant 2310015-051).
  }
  \and
  Parth Mittal\thanks{(\texttt{parth.mittal@uwaterloo.ca})
    University of Waterloo.
    Supported in part by a David R. Cheriton Graduate
    Scholarship and Sepehr Assadi’s Sloan Research Fellowship, an NSERC
    Discovery grant, and a
    startup grant from University of Waterloo.
  }
}

\DeclarePairedDelimiter{\bracket}[]
\DeclarePairedDelimiter{\paren}()

\DeclarePairedDelimiter{\card}{\vert}{\vert}

\DeclarePairedDelimiter{\set}{\{}{\}}

\DeclarePairedDelimiterXPP\lonenorm[1]{}\lVert\rVert{_1}{#1}

\DeclareMathOperator*{\distrib}{\mu}
\DeclareMathOperator*{\support}{supp}
\DeclareMathOperator*{\kldiv}{\mathbb{D}_{\textnormal{\text{KL}}}}
\DeclareMathOperator*{\entropy}{\mathbb{H}}
\DeclareMathOperator*{\inform}{\mathbb{I}}
\DeclareMathOperator*{\expect}{\mathbb{E}}
\DeclareMathOperator*{\variance}{Var}
\DeclareMathOperator*{\covariance}{CoVar}

\newcommand{\II}{\inform}

\DeclarePairedDelimiterXPP{\TVD}[2]{\Delta_{\textnormal{\text{TVD}}}}(){}{#1, #2}
\DeclarePairedDelimiterXPP{\PTD}[2]{\Lambda}(){}{#1, #2}

\DeclarePairedDelimiterXPP{\Ot}[1]{\widetilde{O}}(){}{#1}
\DeclarePairedDelimiterXPP{\Omgt}[1]{\widetilde{\Omega}}(){}{#1}
\DeclarePairedDelimiterXPP{\BigO}[1]{O}(){}{#1}

\NewDocumentCommand{\Prob}{sO{}E{_}{{}}m}{%
  \Pr_{#3}
  \IfBooleanTF{#1}
  {\bracket*{#4}}
  {\bracket[#2]{#4}}
}

\NewDocumentCommand{\Exp}{sO{}E{_}{{}}m}{%
  \expect_{#3}
  \IfBooleanTF{#1}
  {\bracket*{#4}}
  {\bracket[#2]{#4}}
}

\DeclarePairedDelimiterXPP{\Var}[1]{\variance}[]{}{#1}
\DeclarePairedDelimiterXPP{\Cov}[1]{\covariance}[]{}{#1}
\DeclarePairedDelimiterXPP{\eexp}[1]{\exp}(){}{#1}

\NewDocumentCommand{\Dist}{sO{}E{_}{{}}m}{%
  \distrib_{#3}
  \IfBooleanTF{#1}
  {\paren*{#4}}
  {\paren[#2]{#4}}
}

\DeclarePairedDelimiterXPP{\Supp}[1]{\support}(){}{#1}

\DeclarePairedDelimiterXPP{\KL}[2]{\kldiv}(){}{#1 \;\delimsize\|\; #2}

\DeclarePairedDelimiterXPP{\Ent}[1]{\entropy}(){}{#1}
\DeclarePairedDelimiterXPP{\Inf}[2]{\inform}(){}{#1 \; ; \; #2}

\newcommand{\eps}{\varepsilon}

\newcommand{\poly}{\mbox{\rm poly}}

\newcommand{\dpiv}{d^{\pi}_v}
\newcommand{\kpiv}{k^{\pi}_v}

\newcommand{\kest}{\tilde{k}}

\newcommand{\fC}{\mathcal{C}}
\newcommand{\fP}{\mathcal{P}}

\newenvironment{tbox}{\begin{tcolorbox}[
		enlarge top by=5pt,
		enlarge bottom by=5pt,
		 breakable,
		 boxsep=0pt,
                  left=4pt,
                  right=4pt,
                  top=10pt,
                  arc=0pt,
                  boxrule=1pt,toprule=1pt,
                  colback=white
                  ]
	}
{\end{tcolorbox}}

\newcommand{\Language}[1]{\textnormal{\textsf{#1}}}
\newcommand{\DOneC}{\paren{\Delta{} + 1}\textnormal{-coloring}}
\newcommand{\KSetInt}{k\Language{-Slack-Int}}

\newcommand{\NPHard}{\Language{NP}-hard}

\newcommand{\stcomp}[1]{\overline{#1}}

\newcommand{\rv}[1]{\mathsf{#1}}

\newcommand{\rC}{\rv{C}}

\newcommand{\rK}{\rv{K}}

\newcommand{\rS}{\rv{S}}

\newcommand{\rX}{\rv{X}}

\newcommand{\mireal}[1][]{
  \ifx\relax#1\relax%
    \II(\mione \,; \mitwo)%
  \else%
    \II(\mione \,; \mitwo\mid #1)%
  \fi
}

\begin{document}
\pagenumbering{arabic}

\maketitle
\begin{abstract}
  We study the communication complexity of $(\Delta + 1)$ vertex coloring,
  where the edges of an $n$-vertex graph of maximum degree $\Delta$
  are partitioned between two players.
  We provide a randomized protocol which uses $O(n)$ bits of communication
  and ends with both players knowing the coloring.
  Combining this with a folklore $\Omega(n)$ lower bound, this settles the
  randomized communication complexity of $\DOneC$ up to constant factors.
\end{abstract}




\section{Introduction}

Graph coloring is a fundamental problem in computer science.
Given an undirected graph, we are asked to assign each
vertex a color such that no two adjacent vertices have the same color.
Minimizing the number of colors used, even approximately~\cite{LundY94, KhannaLS00},
is known to be \NPHard.
On the other hand, for a graph with maximum degree $\Delta$, there is a simple
greedy algorithm that finds a $\paren{\Delta + 1}$-coloring in linear time.
Due to the sequential nature of this algorithm, the problem has received
a great deal of attention in the sublinear and distributed models of
computation, with recent sublinear time~\cite{AssadiCK19},
semi-streaming space~\cite{AssadiCK19},
distributed~\cite{Linial92,BarenboimEPS16,ChangLP20,GhaffariK21}
massively parallel computation~\cite{ChangFGUZ19,CzumajDP21}, and
dynamic~\cite{BhattacharyaCHN18,BhattacharyaGKLS22,HenzingerP22} algorithms.

We look at this problem in the two-party communication model of~\cite{Yao79}.
The edges of the input graph $G$ are partitioned
between two players, Alice and Bob, who wish to compute some function
(or relation) on $G$ while minimizing the number of bits they send
to each other.
We refer readers to~\cite{KushilevitzN96,RaoY20} for an extensive overview of
the field.
Communication lower bounds have been used to obtain an astonishing breadth
of impossibility results
in distributed computing~\cite{SarmaHKKNPPW11,BachrachCDELP19},
streaming~\cite{IndykW03,ChakrabartiR12},
data-structures~\cite{MiltersenNSW98}, and circuit complexity~\cite{KarchmerW90}.

Several graph problems that have fast classical algorithms
have been studied in the two-party communication model, often leading to
creative algorithms that get very close to somewhat trivial lower bounds.
For example, deciding whether an $n$-vertex graph is connected has an easy
$O(n \log n)$ communication protocol (either player just sends a spanning
forest of their subgraph);
finding the minimum cut has an $O(n \log n)$ communication protocol,
and an $\Omega(n \log \log n)$ randomized communication lower
bound~\cite{AssadiD21};
finding a maximum matching in bipartite graphs has an $O(n \log^2 n)$
communication protocol~\cite{BlikstadBEMN22}.
All of these problems have an $\Omega(n \log n)$ deterministic communication
lower bound~\cite{HMT88}, and an $\Omega(n)$ randomized communication lower
bound via a reduction from set-disjointness.
Closing the gap between this linear lower bound and nearly-linear upper bounds
in the randomized case has been a longstanding open problem
in communciation complexity.

\subsection{Our Contributions}
In this paper, we settle the randomized communication complexity of
$(\Delta + 1)$-coloring.
For this problem, the current state of the art is similar to the aforementioned graph problems.
The best upper bound (to our knowledge) is $O(n \log^2 \Delta)$, alluded
to by~\cite{AssadiCGS23} (see \Cref{rem:det} for a description of this protocol).
It would be natural to expect $\Theta(n \log \Delta)$ to be the
``right answer'' for this problem
--- after all this is the number of bits one would use to write down a
$(\Delta + 1)$-coloring.
However, we show that if Alice has all the edges of the graph, there
is a deterministic $O(n)$-bit protocol with which she can communicate a
$\DOneC$ to Bob.
This is formalized as a non-deterministic communication
upper bound in \Cref{thm:non-det}.
The upshot is that we can hope to beat $n \log \Delta$ in the general case, and
indeed, our main result is a randomized protocol that uses $O(n)$ bits of
communication to find a $\DOneC$:

\medskip

\begin{result}
    \medskip
\begin{restatable}{theorem}{RandUB}\label{thm:main}
    There exists a zero-error
    randomized protocol that given an $n$-vertex graph $G$
    and its maximum degree $\Delta$, finds a $(\Delta + 1)$-coloring of $G$
    using $O(n)$ bits of communication in expectation.
\end{restatable}
\end{result}

Using Markov's inequality,
the protocol of \Cref{thm:main} can be adapted to provide worst-case guarantees
if we allow the algorithm to fail with constant probability.
We also assume players have access to a source of shared randomness.
This assumption can be removed by a classical result of
Newman \cite{Newman91}.
\begin{corollary}
    There exists a private-randomness protocol that given
    an $n$-vertex graph $G$ and its maximum degree $\Delta$, outputs
    a proper $(\Delta + 1)$-coloring of $G$ with probability $2/3$ and
    uses $O(n)$ bits of communication in the worst-case.
\end{corollary}

To complement this result, we show a (folklore) communication lower
bound ruling out constant-error $o(n)$ communication protocols,
via a reduction from the identity function (i.e.\ on an $n$-bit string
$x$, the output is $x$).
We also investigate better bounds on the probability of linear communication.
We show that when $\Delta$ is small compared to $n$, the protocol runs in
$O(n)$ communication with high probability.%
\footnote{With probability at least $1 - n^{-c}$ for any desirably large constant $c > 0$.}
When $\Delta$ is large, we show a slightly weaker $O(n \log^* \Delta)$ bound on
communication holds with high probability.

\subsection{Technical Overview}

\paragraph{The Non-Deterministic Upper Bound.}
We rely on a result of Csikvári~\cite{Zhao17} about the number
of proper $(\Delta + 1)$-colorings of an $n$-vertex graph which tells us that
a uniformly random coloring is proper with probability at least $e^{-n}$.
With a straightforward application of the probabilistic method, this gives
us a set of $2^{O(n)}$ colorings which contains a proper coloring for \emph{any}
$n$-vertex graph, which immediately implies the non-deterministic communication
upper bound in \Cref{thm:non-det} (the prover can just point out a valid coloring).

\paragraph{The Randomized Upper Bound.}
The first insight is that if we run the greedy algorithm on a random
(instead of arbitrary)
order of the vertices, each vertex has (on average) at least $\Delta / 2$ colors
available (instead of just $1$) when we try to color it.
(This is also the core of the proof of Csikvári.)
We can exploit this glut of available colors by sampling random colors
from $\bracket{\Delta+1}$, and
paying $2$ bits of communication per sample to decide if it is valid.
One can show this yields an $O(n \log \Delta)$ communication protocol
(see \Cref{sec:easy}).
However, this is not enough to go all the way; e.g.\ on a $(\Delta + 1)$-clique,
this algorithm uses $\Theta(\Delta \log \Delta)$ bits of communication.

On the other hand, there is a well known
deterministic algorithm~\cite{AssadiCGS23} that can find
an available color for a single vertex in
$O( \log^2 \Delta )$ communication, which we briefly
sketch here (see \Cref{lem:set-int-k1} for more details).
Because the edges of the graph are partitioned, when Alice and Bob try to
color the vertex $v$, the number of colors used by vertices in $N_A(v)$ and $N_B(v)$
\emph{sum} to less than $\Delta + 1$ (as opposed to just their union being
less than $\Delta + 1$).
Hence, at least one part of any partition of $\bracket{\Delta + 1}$ retains
this ``slack'' of $1$,
which allows a binary search to find an available color in
$O(\log^2 \Delta)$ communication.

To get to $O(n)$ communication, we combine the two ideas above.
Consider some vertex $v$ and let $k$ be its number of available colors.
If Alice had all the edges incident to $v$, she could describe
an available color using $O(\log(\Delta/k))$ bits by indicating the index of the first available color in the string of public randomness.
To approach this ideal complexity
when the edges are split between the two players,
we implement a form of palette sparsification:
Alice and Bob restrict themselves to a random color space of size roughly
$\paren{\Delta/k}^2$.
When $k$ is large (i.e., very often for a random permutation), this color
space is much smaller than $\Delta + 1$, but it still retains the critical
``slack'' property required by the binary search protocol
(see \cref{lem:sampleS}).
This yields an $O(\log^2 (\Delta / k))$ communication protocol for coloring
a single vertex, which in turn is enough to get our main theorem
(see \Cref{sec:actual}).


%
%
%


\section{Preliminaries}

\paragraph{Notation.}
Throughout this paper, we will work with input graphs $G=(V,E)$ on the vertex set
$V = \bracket{n} := \set{1, 2, \ldots, n}$, with maximum degree $\Delta$.
Let $N(v)$ denote the neighborhood of $v$ in $G$, and $d_v = \card{ N(v) }$
its degree.
For any integer $q\ge 1$, a $q$-coloring of $G$ is a vector
$C \in \bracket{q}^n$.
We say the coloring is proper if for all edges $\set{u, v} \in E$ we have
$C(u) \neq C(v)$.
Unless stated otherwise, all colorings considered here will be proper.

\paragraph{Model.}
In the classic two-party model, Alice and Bob are given $n$ and $\Delta$.
The edges of a graph $G$ are adversarially partitioned between the two parties.
We will use $E_A$ and $E_B$ to denote the edges given to Alice and Bob
respectively, and $N_A(v)$ and $N_B(v)$ for the neighborhood of $v$
in $E_A$ and $E_B$.
At the end of the protocol both players should agree on the color of each vertex.
We assume both players have access to public randomness.
Recall that an $\eps$-error public-coin protocol can be translated to
an $(\eps + \delta)$-error private-coin protocol while adding only
$O( \log n + \log(1 / \delta) )$ bits to the communication cost~\cite{Newman91}.

The following inequalities are useful to bound averages in the cost of our algorithm.
In words, adding bigger numbers to an average cannot decrease it,
and adding smaller numbers cannot increase it.

\begin{proposition}\label{prop:avg}
  Let $a_1 \leq \ldots \leq a_m$ be a sequence of real numbers, with $m \geq 2$.
  Then for any $1 \leq j < m$,
  \[
    \frac 1j \sum_{k = 1}^j a_k \leq \frac 1m \sum_{k = 1}^m a_k
    \leq \frac 1{m - j} \sum_{k = j + 1}^m a_k.
  \]
\end{proposition}
\begin{proof}
  Note that each term on the left side is at most $a_j$, and each term
  on the right side is at least $a_{j + 1}$, which immediately gives us
  that the left side is smaller than the right.
  To complete the proof, we can rewrite the average of the whole sequence
  $\frac 1m \sum_{k = 1}^m a_k$
  as a convex combination of the two:
  \[
    \frac jm  \paren*{\frac 1j \sum_{k = 1}^j a_k} +
    \frac {m - j}m \paren*{\frac 1{m - j} \sum_{k = j + 1}^m a_k} \ . \qedhere
  \]
\end{proof}

\subsection{Concentration Inequalities}

Throughout the paper we use a handful of classic concentration bounds.
The first one is the multiplicative Chernoff Bound; see~\cite{DubhashiP09}
or~\cite[Theorem 1.10.1]{Doerr20} for more details.

\begin{proposition}[Chernoff Bound]\label{prop:chernoff}
  Let $\rX_1, \rX_2, \ldots, \rX_n$ be independent random variables in $[0,1]$ and $\rX = \sum_{i=1}^n \rX_i$.
  Then, for any $\eps \in (0,1)$,
  \begin{align*}
    \Prob*{ \rX < (1-\eps)\Exp{\rX} } &\le \exp \paren*{ - \frac{\eps^2 \Exp{\rX}}{2} }
    \shortintertext{and for any $\eps > 0$,}
    \Prob*{ \rX > (1+\eps)\Exp{\rX} } &\le \exp \paren*{ - \frac{\eps^2 \Exp{\rX}}{2 + \eps} }  \ .
  \end{align*}
\end{proposition}

When variables are bounded in arbitrary intervals, we use a non-trivial extension of the Chernoff Bound due to Hoeffding.

\begin{proposition}[Hoeffding Bound~\cite{Hoeffding94}]\label{prop:hoeffding}
  Let $\rX_1, \ldots, \rX_n$ be independent random variables and $a_1, \ldots, a_n$, $b_1, \ldots, b_n$ reals such that for each $i\in[n]$, with probability one $a_i \le \rX_i \le b_i$. Then, if $\rX = \sum_{i=1}^n \rX_i$, for any $t > 0$,
  \[
    \Prob*{ \rX > \Exp{\rX} + t } \le \exp \paren*{ -\frac{2t^2}{\sum_{i=1}^{n} (b_i - a_i)^2} } \ .
  \]
\end{proposition}

Finally, when the outcome of an experiment is a function of independent random variables following a certain Lipschitz regularity condition --- namely, changing one variable cannot affect the outcome by too much --- Hoeffding-like concentration can also be obtained.

\begin{proposition}[{Bounded Differences~\cite[Corollary 5.2]{DubhashiP09}}]\label{prop:bounded-diff}
  Let $\rX_1, \ldots, \rX_n$ be arbitrary independent random variables and $f(x_1, \ldots, x_n)$ a real valued function such that whenever $x$ and $x'$ differ in just the $i$-th coordinate, then $|f(x) - f(x')| \le d_i$.
  Then, for all $t > 0$,
  \begin{align*}
    \Prob*{ f(\rX) > \Exp{f(\rX)} + t }
    &\le \exp \paren*{ -\frac{2t^2}{d} } \ ,
  \end{align*}
  where $f(\rX) = f(\rX_1, \ldots, \rX_n)$ and $d = \sum_{i=1}^n d_i^2$.
\end{proposition}


\section{Non-deterministic Upper Bound}

In this section, we will prove our non-deterministic communication upper bound.
In the non-deterministic version of the model, the randomness is replaced
by a prover who has access to all the edges of the graph.
The prover sends a single message to Alice and Bob, who do not communicate with
each other, and must agree on a proper coloring.
In particular, if the prover's message encodes an improper coloring, at least
one of them must reject it.
The cost of the protocol is the length of the prover's message.

\begin{theorem}\label{thm:non-det}
    There exists a non-deterministic protocol that given an $n$-vertex
    graph $G$ with maximum degree $\Delta$, finds a $(\Delta + 1)$-coloring
    of $G$ using $O(n)$ bits of communication.
\end{theorem}

We will show that, for each $n \ge 1$, there is a set of colorings $\fC_{n,\Delta}$ in $\bracket{\Delta + 1}^n$
of size $2^{\BigO{n}}$ such that any $n$-vertex graph with maximum degree $\Delta$ has
a proper coloring in $\fC_{n,\Delta}$.
Note that this immediately gives an $O(n)$-bit non-deterministic communication
protocol, since the prover can just point out a coloring in $\fC_{n,\Delta}$ that works
for the given graph, with Alice and Bob accepting if and only if there are no
monochromatic edges in their respective subgraphs.

Our first observation is that any graph $G$ has many $(\Delta + 1)$ colorings.
More formally:
\begin{lemma}\label{prop:num-colorings}
  The number of $(\Delta + 1)$-colorings of a graph $G$ on $n$ vertices with
  maximum degree $\Delta$ is
  $\geq \paren{\frac{\Delta+1}{e}}^n$.
\end{lemma}
The proof is an easy modification of a result of Csikvári~\cite[Thm 8.3]{Zhao17}
which lower bounds the number of $q$-colorings of $d$-regular graphs for all
$q \geq d + 1$.

\begin{proof}
  Let $\pi$ be a random permutation of $[n]$. We color the vertices in the order of $\pi$.
  For each $v\in V$, let $d_v^{\pi}$ be the number of neighbors of $v$ coming before $v$ in $\pi$.
  When we color $v$, there are at least $\Delta + 1 - d_v^\pi$ choices of colors. Hence, the total
  number of colorings we can obtain by coloring in this order is
  \[ \# \text{colorings from $\pi$} \geq \prod_{v\in V} (\Delta + 1 - d_v^\pi) \ .\]
  After taking the logarithm, the righthand side becomes $\sum_v \log(\Delta + 1 - d_v^\pi)$.
  Now, observe that $d_v^\pi$ is uniformly distributed in $\set{0, 1, \ldots, d_v}$.
  Hence by linearity of expectation,
  $\Exp_{\pi}{ \log \# \text{ colorings for } \pi }$ is at least:
  \[
    \sum_{v\in V} \frac{1}{d_v+1}\sum_{i = 0}^{d_v}\log(\Delta + 1 - i)
    \geq
    \sum_{v\in V} \frac{1}{\Delta + 1}\sum_{i = 0}^{\Delta}\log(\Delta + 1 - i)
    =
    \sum_{v\in V} \frac{1}{\Delta + 1} \log \paren*{ \paren{\Delta + 1}!}
    \geq
    n\log\paren*{\frac{\Delta + 1}e} \ ,
  \]
  where the first inequality is by \Cref{prop:avg}, and the
  second inequality holds because
  $(\Delta + 1)! \geq \paren{\frac{\Delta + 1}{e}}^{\Delta + 1}$
  by Stirling's approximation.
  To get the lemma, we apply Jensen's inequality on $\log x$ to obtain
  that $\log \Exp_{\pi} { \# \text{colorings for } \pi } \geq
  \Exp_{\pi}{ \log \# \text{colorings for } \pi }$.
\end{proof}

Fix a graph $G$ on $n$ vertices, with maximum degree $\Delta$.
\Cref{prop:num-colorings} implies that a uniformly random coloring
from $\bracket{\Delta + 1}^n$ is proper for $G$ with probability at least $1/e^{n}$.
This means that if we construct $\fC_{n,\Delta}$ by sampling $t = e^n \cdot n^2$
colorings at random in $\bracket{\Delta+1}^n$,
the probability that \emph{none} of them is proper for $G$ is at most:
\[
  \paren{1 - 1 / e^n}^t \leq \exp\paren{ -t / e^n } = \exp\paren{ -n^2 }.
\]
Since there are at most $2^{\binom{n}{2}}$ graphs with $n$ vertices (and maximum
degree $\Delta$), a union bound over all such graphs yields that the probability
that $\fC_{n,\Delta}$ does not contain a proper coloring for one of them is at most
$2^{\binom{n}{2}} \cdot e^{-n^2} < 1$.
Hence, there exists a set $\fC_{n,\Delta}$ of $e^n \cdot n^2 \leq 2^{4n}$ colorings in
$\bracket{\Delta + 1}^n$ such that any $n$-vertex graph with maximum degree
$\Delta$ has a proper coloring in $\fC_{n,\Delta}$.


\section{Two Randomized Upper Bounds}

In this section, we will prove our main theorem, which we restate below.

\RandUB*

\subsection{The First Attempt}\label{sec:easy}

We will start by presenting a simple $O(n \log \Delta)$ communication protocol
for the problem, and then show that making one of its subroutines
(even slightly) more efficient results in an $O(n)$ communication protocol.

\begin{Algorithm}\label{alg:easy}
A $O(n \log \Delta)$ communication protocol for $\DOneC$.
\begin{enumerate}
  \item Alice and Bob choose a random permutation $\pi$ of $[n]$ with public
    randomness.
  \item They iterate over the vertices ordered by $\pi$, and on vertex $v$:
    \begin{enumerate}
      \item\label{step:one-shot}
        With public randomness, they sample a color $c$ uniformly at random
        from $\bracket{\Delta + 1}$, and Alice (respectively Bob) sends a bit
        indicating whether $c$ is already assigned to a vertex in $N_A(v)$
        (respectively $N_B(v)$).
      \item If they both send $0$ (i.e.\ $c$ is used neither in $N_A(v)$
        nor
        $N_B(v)$)
        they assign $c$ to $v$, and move on to the next vertex.
        Otherwise, they go to \Cref{step:one-shot}.
    \end{enumerate}
\end{enumerate}
\end{Algorithm}

\begin{lemma}\label{lem:easy-alg}
  \Cref{alg:easy} uses $O(n \log \Delta)$ bits of communication in expectation.
\end{lemma}

\begin{proof}
Before we analyze the algorithm, we need to set up some notation.
For a vertex $v$ and a permutation $\pi$, let
$\dpiv$ denote the number of neighbors of $v$ in $G = (V, E_A \cup E_B)$
that appear before $v$ in $\pi$, and $\kpiv := \Delta + 1 - \dpiv$
(think of $\kpiv$ as a lower bound on the number of colors that will
be available when we try to color $v$).

\noindent
First, we will use the randomness of \Cref{step:one-shot}:
Observe that on each sample in \Cref{step:one-shot}, the probability
of sampling a valid color for $v$ is at least $\kpiv / (\Delta + 1)$.
The number of trials required to find a valid color is a geometric
random variable, and hence its mean is upper bounded by $(\Delta + 1) / \kpiv$.
Since Alice and Bob send $1$ bit each per trial, the communication cost of
coloring $v$ is $\leq 2 (\Delta + 1) / \kpiv$ in expectation.

\noindent
Next, we will use the randomness of $\pi$:
Note that $\kpiv$ is uniformly random from the set
$\set{\Delta + 1 - d_v, \Delta + 2 - d_v, \ldots , \Delta + 1}$.
Because the analysis of \Cref{step:one-shot} holds for any permutation $\pi$,
the expected cost of coloring $v$ is upper bounded by:
\[
  \frac 1{d_v + 1}
  \sum_{i = \Delta + 1 - d_v}^{\Delta + 1} \frac{2\paren{\Delta + 1}}{i}
  \leq
  \frac 1{\Delta + 1}
  \sum_{i = 1}^{\Delta + 1} \frac{2\paren{\Delta + 1}}{i}
  = 2H_{\Delta + 1} = O(\log \Delta),
\]
where the first inequality is by \Cref{prop:avg},
and $H_m$ is the $m$-th harmonic number.
Summing up over all vertices, by linearity of expectation, the expected
communication cost of the algorithm is $O(n \log \Delta)$.
\end{proof}

\begin{remark}\label{rem:non-disj}
  Note that this algorithm also works if $E_A$ and $E_B$ are not disjoint,
  unlike the algorithm in our main result.
\end{remark}

\subsection{A Haystack With Many Needles}

The analysis of \Cref{alg:easy} is tight when $G$ is a collection of
disjoint $(\Delta + 1)$-cliques.
This is because we are spending $2$ bits of communication per sampled color,
and to even sample all $\Delta + 1$ colors, we need $\Omega(\Delta \log \Delta)$
samples per clique by Coupon Collector.
Now suppose we have an ``ideal'' vertex $v$, where Alice has all the edges in
$G$ incident to $v$.
Then Alice could just point out the first sampled color that works for $v$
in $O( \log ( \Delta / \kpiv ) )$ bits of communication, instead of the
$O(\Delta / \kpiv)$ we spent.
To get (close to) this ideal, we abstract coloring a single vertex into the
following set-intersection type problem:

\begin{problem}[$\KSetInt$]\label{prob:set-int}
  Alice and Bob are given sets $X$ and $Y \subsetneq \bracket{m}$ respectively,
  such that $ \card{X} + \card{Y} \leq m - k$ for some $k \geq 1$.
  Both of them are also given the integer $m$,
  but neither of them knows $k$.
  They wish to find an element in the intersection of $\stcomp{X} := [m] \setminus X$ and
  $\stcomp{Y} := [m] \setminus Y$.
\end{problem}

To see why it is relevant, we recast coloring a single vertex as an instance of
\Cref{prob:set-int}.
As before, fix some vertex $v$ and permutation $\pi$ in \Cref{alg:easy}.
Let $A$ and $B$ denote the sets of colors used so far in $N_A(v)$ and $N_B(v)$
respectively; then coloring $v$ is equivalent to solving \Cref{prob:set-int}
with $m = \Delta + 1$, $X = A$, $Y = B$ and $k = \kpiv$.
We emphasize that the promise $\card{X} + \card{Y} \leq m - k$ holds because sets
$N_A(v)$ and $N_B(v)$ are disjoint (the edges of $G$ are \emph{partitioned}
between Alice and Bob).
The main lemma of this section gives an efficient algorithm for this problem.

\begin{lemma}\label{lem:set-int}
  There exists a randomized protocol that solves $\KSetInt$
  in $O( \log^2 (m / k) + 1 )$
  bits of communication in expectation,
  and $O( \log^3 m )$ bits of communication in the worst-case.
\end{lemma}

Note that the promise on $\card{X} + \card{Y}$ is crucial.
For example, if we only knew $\card{X \cup Y} < m$, then this is just the
set-intersection problem, which is known to require $\Theta(m)$ bits of
communication~\cite{KalyanasundaramS92}.

In this language, the coloring step in \Cref{alg:easy} was taking random
samples in $\bracket{m}$, succeeding in each step with probability at least
$k / m$, for an expected communication cost of $2m / k$.
To get a better algorithm, we will first focus on the hardest case for this
random sampling algorithm (i.e. $k = 1$).
Using the promise of $\KSetInt$, we can binary search for an element.

\begin{lemma}\label{lem:set-int-k1}
  There exists a deterministic protocol  that solves $\KSetInt$
  in $\BigO{ \log^2(m) }$ bits of communication, for any $k \geq 1$.
\end{lemma}
\begin{proof}
  The key observation is that we can binary search for the target element.
  Let $L = \bracket{m / 2}$ and $R = \bracket{m} \setminus L$, and consider
  $\card{X \cap L} + \card{Y \cap L}$ and
  $\card{X \cap R} + \card{Y \cap R}$.
  The two numbers add up to $\card{X} + \card{Y}$, which is smaller than $m$.
  Hence $\card{X \cap L} + \card{Y \cap L}$ must be smaller than $m / 2$, or
  $\card{X \cap R} + \card{Y \cap R}$ smaller than $m - m / 2$, giving a smaller
  instance of the same problem.
  Alice and Bob can send $\card{X \cap L}$ and $\card{Y \cap L}$
  respectively, in $\leq 2 \log m$ bits, and then recurse on the correct half.
  This gives a protocol for finding an element in $\stcomp{X} \cap \stcomp{Y}$
  in $O( \log^2 m )$ communication.
\end{proof}

\begin{remark}\label{rem:det}
  Note that \Cref{lem:set-int-k1} gives an $\BigO{ n \log^2 \Delta }$
  deterministic communication protocol for $(\Delta + 1)$-coloring, since
  coloring each vertex (ordered by an arbitrary permutation) is a
  $\KSetInt$ instance (but with no guarantee on $k$ other than $k \geq 1$).
\end{remark}

The upshot of the binary search protocol is that the condition
$\card{X} + \card{Y} < m$ can be preserved while recursing on smaller
subproblems.
However, this costs too much communication, and we are not exploiting the
fact that we have a lot of slack ($\card{X} + \card{Y} \leq m - k$, not just
$m - 1$).

\noindent
We know that if we sample $m / k$ elements from $\bracket{m}$, we expect to
see at least one element from $\stcomp{X} \cap \stcomp{Y}$.
And indeed, if we let $S$ denote the set of sampled elements, we are trying
to solve \Cref{prob:set-int} on $S \cap X$ and $S \cap Y$, with the
caveat that $\card{S \cap X} + \card{S \cap Y}$ may not be $\leq \card{S} - k$
(or for that matter even $< \card{S}$).
It turns out that with roughly $m^2 / k^2$ samples in $S$, we get
$\card{S \cap X} + \card{S \cap Y} < \card{S}$ with $\Omega(1)$
probability.
This means that we can run the binary search protocol of
\Cref{lem:set-int-k1} on this smaller instance instead, and pay only
$O( \log^2 ( m^2 / k^2 ) ) = O( \log^2 (m / k) )$ bits of communication.

A downside of the sketch above is that we need to know $k$ (to compute $p$).
This turns out to not be an issue, since any guess $\kest \leq k$ suffices
(as we will see in the following lemma), and we can arrive at such a $\kest$ in
$O( \log ( m / k ) )$ guesses starting from $m$ (as we will see after the lemma).

\begin{lemma}\label{lem:sampleS}
  Let $X, Y \subsetneq \bracket{m}$ be a $\KSetInt$ instance, and let
  $p = \min\set{150m / \kest^2, 1}$, for some $\kest \leq k \leq m$.
  Define the random set $S$ by sampling each element of $\bracket{m}$
  independently with probability $p$. Then
\[ \Prob[\Big]_{S} { \card{S \cap X} + \card{S \cap Y} \geq \card{S} } \leq
   1/2. \]
\end{lemma}
\begin{proof}
  For brevity, let $s := \card{S}$, $s_x := \card{S \cap X}$ and
  $s_y := \card{S \cap Y}$.
  By linearity of expectation
  $\Exp{s}$ is $m \cdot p$, and
  $\Exp{ s_x + s_y } \leq \paren{m - k} \cdot p$.
  This means that $s_x + s_y < s$ as long as:
  \begin{itemize}
    \item $s_x \leq \Exp{s_x} + k \cdot p / 5$,
    \item $s_y \leq \Exp{s_y} + k \cdot p / 5$,
    \item $s \geq \Exp{s} - k \cdot p / 5$.
  \end{itemize}
  (Because the first two inequalities imply
  $s_x + s_y
  \leq \Exp{s_x + s_y} + 2kp/5
  \leq \Exp{s} - 3kp/5$
  while the last imply that $s$ is strictly larger.)

  Let $\eps = k \cdot p/(5\Exp{ s_x })$.
  When $\eps \leq 1$, the Chernoff bound from \Cref{prop:chernoff}
  on $s_x$ gives
  \begin{align*}
    \Prob*{ s_x \geq (1 + \eps) \cdot \Exp{s_x} }
    &\leq \exp \paren*{ - \frac{\eps^2 \Exp{s_x} }{2 + \eps}}
    \leq \exp \paren*{ - \frac{\eps^2 \Exp{s_x} }{3} }.
    \shortintertext{Plugging in the value of $\eps$,}
    \Prob*{ s_x \geq \Exp{s_x} + k \cdot p / 5}
    &\leq \exp \paren*{ - \frac{\paren{k \cdot p}^2 \Exp{s_x} }{75 \Exp{s_x}^2}}
    \leq \exp \paren*{ - \frac{\kest^2 p^2}{75 \Exp{s_x}}}. \tag{$k \geq{} \kest$}
    \shortintertext{We can afford to be lazy and upper bound $\Exp{s_x}$ by
    $m \cdot p$, and get:}
    \Prob*{ s_x \geq \Exp{s_x} + k \cdot p / 5}
    &\leq \exp \paren*{ - \frac{\kest^2 \cdot p^2}{75 m \cdot p}}
    = \exp \paren*{ - \frac{\kest^2 \cdot p}{75 m}}\\
    &\leq \exp \paren*{ - \frac{150\kest^2 \cdot m}{75 m \cdot \kest^2}}
    = e^{-2}.
  \end{align*}
  If $\eps = k \cdot p / (5 \Exp{s_x}) > 1$, \Cref{prop:chernoff} gives%
\footnote{Using that $\eps^2 / (2 + \eps) = \eps / (2/\eps + 1) \geq \eps/3$ when $\eps \geq 1$.}
\[
  \Prob*{ s_x \geq (1 + \eps) \cdot \Exp{s_x} }
    \leq \exp \paren*{ - \frac{\eps \Exp{s_x}}{3} }
    \leq \exp \paren*{ - \frac{ \kest p }{15} }
    \leq \exp \paren*{ - \frac{150 m}{15 \kest} } \leq e^{-10}.
    \tag{$\kest \leq k \leq m$}
\]
  We can repeat the same argument for $s_y$ and $s$ to obtain the
  lemma.
\end{proof}

\begin{remark}\label{rem:strangeP}
  The choice of the sampling probability $p$ in \Cref{lem:sampleS} may seem
  strange at first, and we attempt to explain it here.
  Since we are paying $O( \log^2 (\card{S}))$ bits of communication
  for the binary search, we want the size of $S$ to be $\poly(m / k)$.
  This restricts us to $p \in \set{1 / k, m / k^2, \ldots}$.
  The most natural choice in this sequence, $p = 1 / k$, makes it
  so that the gap between $\Exp{ s_x + s_y }$ and $\Exp{ s }$ is only a
  constant, whereas the expectations themselves are $\Omega(m / k)$, which
  means we cannot hope for $s_x + s_y < s$ with significant probability.
  The second most natural choice makes the expectations $O(m^2 / k^2)$
  while making the gap $\Omega(m / k)$, which is exactly the area where
  a binomial random variable is concentrated.
\end{remark}
Using the lemma, we have the following algorithm for $\KSetInt$:

\begin{Algorithm}\label{alg:set-int}
  A $\BigO*{ \log^2 (m / k) }$ communication protocol for $\KSetInt$.
  \begin{enumerate}
    \item[\textbf{Input}:] Alice gets a set $X \subsetneq \bracket{m}$, Bob gets a
      set $Y \subsetneq \bracket{m}$ such that $\card{X} + \card{Y} \leq m - k$.
    \item[\textbf{Output}:] Any element from $\stcomp{X} \cap \stcomp{Y}$.
    \item For $\kest = m, m / 2, \ldots , 1$ (a sequence of
      exponentially decreasing guesses for $k$):
    \begin{enumerate}
      \item Alice and Bob choose $S$ by sampling each element of $\bracket{m}$
        independently with probability $p = \min\set{1, 150 m / \kest^2}$,
        using public randomness.
      \item\label{step:size-test}
        They test if $\card{S \cap X} + \card{S \cap Y} < \card{S}$.
        If not, they continue to the next value of $\kest$, otherwise:
      \item\label{step:bin-search}
        They run the binary search protocol (\cref{lem:set-int-k1}) to find an element in
        $S \setminus \paren*{ X \cup Y }$,
        and return it as the answer.
    \end{enumerate}
  \end{enumerate}
\end{Algorithm}

Note that the algorithm always terminates with a correct answer, since for
$\kest$ small enough, $p = 1$, and we just run the binary search on
all of $\bracket{m}$.
Hence \Cref{alg:set-int} never uses more than
$O(\log^3 m)$ bits of communication.
To finish the proof of \Cref{lem:set-int}, we bound the expected communication cost
of \Cref{alg:set-int}:

\begin{proof}[Proof of \Cref{lem:set-int}]

Observe that we spend $O( \log (m / k) )$ iterations before
$\kest \leq k$ and we can apply \Cref{lem:sampleS}.
On each of these iterations, the test in \Cref{step:size-test} costs only
$\BigO{\log (m / k)}$ bits of communication, since $\kest > k$ for all of them.
If we get lucky, and the test passes, the binary search takes
$\BigO{ \log^2(m / k) }$ bits, and we are done.

Otherwise, once $\kest \leq k$, the cost of each iteration increases beyond
$\log (m / k)$, but the probability of reaching these values of $\kest$ drops
exponentially.
By applying \Cref{lem:sampleS} inductively, we obtain that the probability
that $\kest \leq k / 2^i$ when the algorithm terminates is at most $1/2^i$.
For a fixed $\kest$, the expected cost of \Cref{step:size-test} is $\Exp_S{\log|S|} \leq \log\Exp_S{|S|} = O(\log(m/\kest))$ communication.
This is dominated by the cost of \Cref{step:bin-search}, which takes
$O(\log^2(m/\kest))$ bits of communication in expectation.
So, the expected communication cost of
\Cref{alg:set-int}
after $\kest \leq k$ is (up to a constant factor) upper bounded by:
\[
  \sum_{i \geq 0} \log^2 (2^{i} \cdot m / k) \cdot 2^{-i}
  =
  \sum_{i \geq 0} \paren*{ i^2 + \log^2 (m / k) + 2i\log(m / k)}
  \cdot 2^{-i}
  = O( \log^2 ( m / k ) + 1 ) \ . \qedhere
\]
\end{proof}

\subsection{Stitching Things Together}\label{sec:actual}

To get our $O(n)$-bit communication protocol for $\DOneC$, we simply plug
\Cref{alg:set-int} into the coloring steps of \Cref{alg:easy}.

\begin{Algorithm}\label{alg:actual}
An $O(n)$ communication protocol for $\DOneC$.
\begin{enumerate}
  \item Alice and Bob choose a random permutation $\pi$ with public randomness.
  \item They iterate over the vertices ordered by $\pi$, and on vertex $v$,
    they run \Cref{alg:set-int} with $m = \Delta + 1$, $X$ (respectively $Y$)
    equal to the set of colors used in $N_A(v)$ (respectively $N_B(v)$).
\end{enumerate}

\end{Algorithm}

\begin{proof}[Proof of \Cref{thm:main}]
As before, for a vertex $v$ and permutation $\pi$, let $\dpiv$ denote the
degree of $v$ among its predecessors in $\pi$, and $\kpiv$ denote
$\Delta + 1 - \dpiv$.
By \Cref{lem:set-int}, \Cref{alg:set-int} colors $v$ in
$\BigO{ \log^2 \paren{(\Delta + 1) / \kpiv} }$ bits of communication.
To finish proving \Cref{thm:main}, we will use the randomness of $\pi$ to
show that this quantity is a constant in expectation.

Since $\pi$ is a uniformly random permutation, $\kpiv$ is uniformly random over
$\set{\Delta + 1 - d_v, \ldots , \Delta + 1}$, and
the expected cost of coloring $v$ is (for some constant $c$):
\[
  \frac 1{d_v + 1}
  \sum_{i = \Delta + 1 - d_v}^{\Delta + 1} c\log^2\paren*{\frac {\Delta + 1} i}
  \leq
  \frac 1{\Delta + 1} \sum_{i = 1}^{\Delta + 1}
  c\log^2\paren*{\frac {\Delta + 1} i}
  \leq
  \frac 1{\Delta + 1}
  \sum_{i = 1}^{\Delta + 1} 5c \cdot {\sqrt{\frac {\Delta + 1} i}},
\]
where the first inequality is by \Cref{prop:avg}, and
the second inequality holds because $\log^2(x) \leq 5\sqrt{x}$ for all
$x \geq 1$.
Then by using the standard integral upper bound, the sum from $i=2$ to $\Delta+1$
is at most:
\[
  \frac 1{\Delta + 1} \int_{1}^{\Delta + 1} 5c\cdot \sqrt{\frac {\Delta + 1} x}
  \cdot \, dx
  \leq
  \frac {5c}{\sqrt{\Delta + 1}} \int_{1}^{\Delta + 1} \frac 1 {\sqrt x} \cdot \, dx
  \leq
  \frac {5c}{\sqrt{\Delta + 1}} \cdot 2 \sqrt{\Delta + 1} = O(1).
\]
Since the first term of the sum is $\leq 5c = O(1)$, the cost of coloring $v$ is $O(1)$ in expectation.
\end{proof}

\begin{remark}
  The analysis above goes through as long as we solve a $\KSetInt$ instance
  in $\paren{m / k}^{1 - \eps}$ bits of communication for any $\eps > 0$.
  This may be useful for other problems where the ``coloring a single vertex''
  task does not admit a very efficient algorithm.
\end{remark}

\subsection{High Probability Analysis}

To complete our analysis, we investigate the probabilistic guarantees given by \Cref{alg:actual}. We first show by standard concentration techniques that when $\Delta$ is small compared to $n$, our algorithm uses $O(n)$ bits with high probability. Formally,

\begin{theorem}
  \Cref{alg:actual} communicates $O(n)$ bits with probability at least
  $1- 2\exp\paren*{ - \frac{n}{9\Delta^2\log^2(\Delta+1)} }$.
\end{theorem}

\begin{proof}
Sample independent random variables $\rX = (\rX_1, \rX_2, \ldots, \rX_n)$ uniformly in $(0,1)$. With probability one, all $\rX_i$'s are different and induce the permutation $\pi(v) = |\set{u\in [n]: \rX_u < \rX_v}| + 1$ (the plus one is to have the smallest index at one). The permutation thus obtained is uniform and we can see the permutation sampled in \Cref{alg:actual} as a function of the $\rX_i$'s.
For each vertex $v$, define $\rK_v = \rK_v(\rX) := \Delta + 1 - \card{\set{u\in N(v): \rX_u < \rX_v}}$.
We first show $\rS = \rS(\rX) = \sum_v \log^2\paren{\frac{\Delta+1}{\rK_v}}$ is concentrated around $O(n)$. We will then argue that with high probability, \Cref{alg:actual} uses $O(\rS)$ communication.

We use the method of bounded differences. Let $i\in[n]$ and $x, x' \in \paren{0,1}^n$ be two vectors which differ only in the $i$-th coordinate. Let $\rS(x)$ and $\rS(x')$ be the value of the sum $\rS$ when $\rX=x$ and $\rX=x'$.
Compared to $x$, changing $x_i$
may increase or decrease $\rK_i$ by up to $\Delta$, and
increase or decrease $\rK_v$ by up to $1$ for any neighbor $v$ of $i$.
To see how this affects $\rS = \sum_v \log^2(\frac{\Delta+1}{\rK_v})$, we bound
  the change in $\log^2( \frac{\Delta + 1}{\rK_i} )$ by $\log^2\paren{\Delta + 1}$;
  for the neighbors of $i$, we need to bound
  $\card{\log^2 \paren{ \frac{\Delta+1}{k} } - \log^2 \paren{ \frac{\Delta+1}{k'} }}$
for every pair of integers $1 \leq k, k' \leq \Delta + 1$
such that $\card{k - k'} \le 1$.
The function $g(y) = \log^2 \paren{ \frac{\Delta+1}{y} }$ is differentiable
when $1 \leq y \leq \Delta + 1$, and
\[
  \card{g'(y)} \leq
  {2\log\paren*{ \frac{\Delta + 1}y}} / y
  \leq
  2\log \paren{\Delta + 1}.
\]
Using the mean value theorem\footnote{The mean value theorem implies that for a differentiable $g:[a, b] \to \mathbb{R}$, $|g(a) - g(b)| \leq |b - a| \sup_{x \in (a,b)} |g'(x)|$. So here, $|g(k) - g(k')| \leq \sup_{x\in (k, k')} |g'(x)| \leq 2\log(\Delta+1)$.}, this implies $\card{g(k) - g(k')} \le 2\log(\Delta+1)$.

We get that
$
\card{\rS(x) - \rS(x')} \leq \Delta \cdot 2\log(\Delta+1) + \log^2(\Delta + 1)
\leq 3\Delta \cdot \log(\Delta + 1)$.
Since $\rS$ is a function of the $n$ independent random variables $\rX_1, \ldots, \rX_n$, \Cref{prop:bounded-diff} implies
\[
   \Prob { \rS > \Exp \rS + n }
   \le
   \exp \paren*{ -\frac{2n^2}{n\paren{ 3\Delta \log(\Delta+1) }^2 } }
   \le
   \exp \paren*{ - \frac{n}{9\Delta^2\log^2(\Delta+1)}} =: p.
 \]
Since $\Exp_\pi{\rS} = O(n)$ (see the proof of \Cref{thm:main}),
we have that $\rS = O(n)$ with probability at least $1 - p$.
We now bound the communication of the algorithm.
We henceforth assume that $\pi$ is fixed and, in particular, $\rS$ and $\rK_v$ for every $v$ are fixed values. Let $\rC_v$ denote the communication cost of coloring $v$ with \Cref{alg:set-int}.
By \cref{lem:set-int}, $\Exp{ \rC_v } \leq O(\log^2(\frac{\Delta+1}{\rK_v})+1)$, where this expectation is only over the random colors chosen by
Alice and Bob while coloring $v$, and \emph{not} the permutation $\pi$, and $\rC_v \leq O(\log^3 \Delta)$ in the worst-case.
So the total cost of the algorithm
$\rC := \sum_v \rC_v$ has expectation $O(\rS + n)$.
As the variables $\set{\rC_v}_v$ are mutually independent, the
Hoeffding bound (\Cref{prop:hoeffding} with $a_v = 0$ and $b_v = O(\log^3 \Delta)$) on variables $\rC_v$ implies
\[
  \Prob*{ \rC > \Exp{\rC} + n  } \le
  \exp \paren*{ - \Omega\paren*{ \frac{n}{\log^6 \Delta} } } \leq p,
\]
where the last inequality holds when $\Delta$ is greater than some constant.\footnote{When $\Delta = O(1)$ the protocol uses $O(n\log^3 \Delta) = O(n)$ communication in the worst-case.} By union bound, \Cref{alg:actual} uses $\rC = O(\rS + n) = O(n)$ communication with probability at least
$1 - 2p$.
\qedhere

\end{proof}

The reason we fail to provide high probability guarantees when $\Delta$ is large is that, when sampling a uniform permutation, changing the position of a single element may affect the cost of many vertices. By revealing the permutation ``in batches'', we provide better probabilistic guarantees for large $\Delta$ at the cost of a slightly weaker bound on communication. Let $\log^{(i)} n$ be the $i$-th iterated logarithm, and recall $\log^* n$ is defined as the number of times we can take the logarithm of $n$ before reaching 1.

\begin{theorem}
  When $\Delta \ge \frac{\sqrt{n}}{c\log^3 n}$,
  \Cref{alg:actual} communicates $O(n\log^*\Delta)$ bits w.p.\ at least $1-e^{-\Theta\paren*{ \frac{\sqrt{n}}{\log^5 n} } }$.
\end{theorem}

\begin{proof}
  Let $r = \log^* \Delta - 3$. For each $i\in[r]$, define $V_i$ inductively as the last $n_i = \frac{n}{\paren{\log^{(i)} \Delta}^2}$ vertices in $V \setminus (V_1 \cup \ldots \cup V_{i-1})$ according to $\pi$. For brevity, write $V_{\le i} = V_1 \cup \ldots \cup V_i$ (and similarly for $<i$, $\ge i$, $> i$). We also define $V_{r+1} = V \setminus V_{\le r}$. The algorithm colors all vertices in $V_{r+1}$ first, then those in $V_{r}$ and so on and so forth until it finishes by coloring $V_1$.
  The crux of the proof is to show the following property --- which we refer to as $\fP$ --- about sets $V_i$ holds with high probability:
  \begin{quote}
    For each $i\in [r]$ and $v\in V \setminus V_{\le i}$ with
      $d_v \ge \Delta/2$\ we have
    $\card{N(v)\cap V_{\le i}} \ge
    \frac{\Delta}{8\paren{\log^{(i)} \Delta}^2} := k_i$.%
  \end{quote}

  Before proving $\fP$, we assume it holds and show the lemma.
  Let $k_0=1$.
  Let us first bound the cost of coloring vertices with $d_v \leq \Delta/2$. Since they always have $k^\pi_v = \Delta+1 - d_v \geq \Delta/2$, \Cref{alg:set-int} colors them with $O(1)$ communication in expectation and never more than $O(\log^3 \Delta)$ in the worst-case (\cref{lem:set-int}). The Hoeffding Bound (\Cref{prop:hoeffding}) shows that the cost of coloring low-degree vertices is $O(n)$ with probability at least $1 - \exp\paren*{ -\Omega\paren*{\frac{n}{\log^6 \Delta}} }$. We assume henceforth $d_v\ge\Delta/2$.
  Observe that by $\fP$, for all $i \in \bracket{r + 1}$
  each $v\in V_{i}$ has $k^\pi_v \ge k_{i-1}$.
  In particular, coloring $V_i$ takes $O(n_i \log^2( \frac{\Delta+1}{k_{i-1}} )) = O(n)$ communication in expectation.
  Since the cost of coloring every single vertex is bounded by $O(\log^3\Delta)$ in the worst-case, the Hoeffding Bound implies the total communication for coloring $V_i$ is $O(n)$ w.p.\ $1-\exp\paren*{ -\Omega\paren{\frac{n}{\log^6 \Delta}}}$.
  By union bound, communication exceeds $O(nr)$ with probability at most $r\exp\paren*{ -\Omega\paren{ \frac{n}{\log^6 n} } } \leq \exp\paren*{ -\Omega\paren{ \frac{n}{\log^6 n} } }$.
  Hence, coloring all vertices takes $O(nr) = O(n\log^* \Delta)$ bits.

  Finally, we prove~$\fP$, which concludes the proof.
  Fix some $i\in[r]$ and $v\notin V_{\le i}$ with $d_v \ge \Delta/2$. We show it has at least $k_i$ neighbors in $V_{\le i}$ with high probability, and the result follows by union bound.
  Fix $V_{<i}$ such that
  $v$ has fewer than $k_i$ neighbors in $V_{<i}$
  (otherwise $v$ already satisfies $\fP$).
  The set $V_i$ is constructed by sampling $n_i$ vertices \emph{without replacement} from the set $V \setminus V_{<i}$ of size at most $n$.
  Note that $v$ has $\card{N(v) \setminus V_{<i}}$ neighbors which
  can be sampled in $V_i$, which is
  $M := d_v - \card{N(v) \cap V_{<i}} \ge \Delta/2 - k_i \ge \Delta/4$ vertices.
  Let $\rK$ count the number of neighbors of $v$ sampled in $V_i$.
  We expect $\Exp{\rK} = \frac{M}{n}n_i \ge 2k_i$.
  The samples are not independent, since $V_i$ is sampled without replacement,
  but as Hoeffding showed in~\cite[Theorem 4]{Hoeffding94} (or \cite[Section 1.10.2.3]{Doerr20}), the Chernoff Bound
  for variables sampled with replacement can be transferred to variables without replacement.
  Hence, using \Cref{prop:chernoff},
  \[
    \Prob*{ \rK < k_i } \leq
    \Prob*{ \rK < \Exp{\rK}/2 }
    \le
    \exp \paren*{ - \Exp{\rK}/8 }
    \le
    \exp \paren*{ - \frac{\sqrt{n}}{64 c\log^5 n} } \ ,
  \]
  where the last inequality holds from the assumption that $\Delta$ is large.
\end{proof}

\section{Lower Bound}

In this section, we provide a simple lower bound matching our upper bounds.

\begin{theorem}\label{thm:lower-bound}
    Any constant-error randomized protocol for computing a $\Delta+1$-coloring
    on $n$-vertices graphs requires $\Omega(n)$ communication in the worst-case.
\end{theorem}

The proof is by a simple reduction from the problem of sending an
$n$-bit string.
Alice constructs a degree-2 graph $G$ with $4n$ vertices such that if Bob knows
a proper 3-coloring of $G$, he can recover $x$.
We remark that in this construction, Alice has all the edges.
Hence, it also lower bounds the non-deterministic complexity of
$\Delta+1$-coloring.

\begin{proposition}\label{prop:send-rand-string}
    Suppose Alice is given a uniformly random string $x \in \set{0,1}^n$ and
    Bob wants to learn $x$.
    Any deterministic protocol where Bob can recover $x$ with probability
    $1/2$ must communicate $\Omega(n)$ bits.
\end{proposition}

The proof of \Cref{prop:send-rand-string} is easy, for example
by using Fano's inequality \cite[Theorem 2.10.1]{CoverT01} to show that
the mutual information between the transcript of any constant error
deterministic protocol and the random string $x$ must be $\Omega(n)$.
By Yao's minimax principle, a distributional lower bound
against deterministic protocols implies a worst-case lower bound against
randomized protocols.

We describe the reduction to complete the proof of \Cref{thm:lower-bound}.

\begin{figure}[ht]
  \caption{The gadget encoding a single bit.
    The dashed red edges are present when the bit is $0$, and the dotted blue
    edges are present when the bit is $1$.
  The solid black edges are always present. }\label{fig:reduction}
  \centering
  \begin{tikzpicture}
    \tikzset{vertex/.style = {shape=circle,draw,minimum size=1.5em}}
    \node[vertex] (1) {1};
    \node[vertex,right=of 1] (2) {2};
    \node[vertex,below=of 2] (3) {3};
    \node[vertex,left=of 3] (4) {4};

    \draw (1) -- (3);
    \draw (2) -- (4);

    \draw[dashed, red] (1) -- (2);
    \draw[dashed, red] (3) -- (4);

    \draw[dotted, blue, thick] (1) -- (4);
    \draw[dotted, blue, thick] (2) -- (3);
  \end{tikzpicture}
\end{figure}
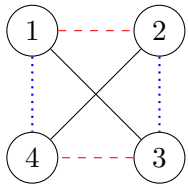

\begin{proof}[Proof of \Cref{thm:lower-bound}]
    Suppose Alice is given a uniformly random string $x \in \set{0, 1}^{n}$.
    Let $v_1, \ldots, v_{4n}$ be the vertices of $G$.
    Gadget $i \in \bracket{n}$ uses vertices
    $v_{4(i-1)+1}, v_{4(i-1)+2}, v_{4(i-1)+3}, v_{4(i-1)+4}$.
    We describe two graphs $H_0$ and $H_1$ on four vertices
    $V_H = \set{1, 2, 3, 4}$.
    For each $i\in \bracket{n}$, Alice constructs a copy of $H_{x_i}$ by mapping
    each vertex $j\in \bracket{4}$ to $v_{4(i-1) + j}$.
    Clearly, if Bob can deduce whether $H = H_0$ or $H = H_1$ from a proper
    coloring of $H$, he can recover $x$ from a proper coloring of $G$.
    \Cref{prop:send-rand-string} implies Alice must send $\Omega(n)$ bits to
    Bob for the protocol to succeed with constant probability.
    We now describe the gadget graphs $H_0, H_1$ (see also \Cref{fig:reduction}).

    Let $x \in \set{0,1}$.
    We describe edges of $H_x$ on vertex set $\bracket{4}$.
    For any value of $x$, add edges $\set{1, 3}$ and $\set{2, 4}$.
    If $x = 0$, put edges $\set{1, 2}$ and $\set{3, 4}$.
    Otherwise, if $x = 1$, put edges $\set{1, 4}$ and $\set{2, 3}$.
    Let $C$ be a proper 3-coloring of $H_0$.
    Either $C(1) = C(4)$ or $C(2) = C(3)$.
    Indeed, if $C(1) \neq C(4)$, then only one out of three colors remain
    available to $2$ and $3$.
    On the other hand, no proper coloring of $H_1$ can have $C(1) = C(4)$ or
    $C(2) = C(3)$.
    Hence by checking these two equalities, Bob can deduce $x$ from a proper
    coloring of $H_x$, which concludes the proof.
\end{proof}


\section{Open Problems}

The most immediate question from our work is to ask if players need to be given $\Delta$ as part of the input.
If we had an $O(n)$-bit communication protocol for computing $\Delta$, we could
remove the requirement of knowing $\Delta$ in \Cref{thm:main}.
However, we know that the communication complexity of this problem is
$\omega(n)$.
Consider the following problem: Alice (respectively Bob) is given as input
the sequence of integers $a_1, \ldots , a_n \in \bracket{n}^n$
(resp. $b_1, \ldots , b_n$), and they wish to compute $\max_i \set{a_i + b_i}$.
This problem can be reduced to finding the maximum degree of a graph on
$n + 2n$ vertices as follows:
For each $i \in \bracket{n}$, Alice (resp. Bob) adds $a_i$ edges from $i$
to vertices in $\set{n + 1, \ldots , 2n}$ (resp. $\set{2n + 1, \ldots , 3n}$),
starting one vertex ahead of where the last edge was added, cycling at
$2n$ (resp. $3n$).
Note that verifying a guess $x$ for the answer to this problem is equivalent
to deciding if there is an $i$ such that $a_i \geq x - b_i$, which has a well
known lower bound of $\Omega(n \log \log n)$
(see~\cite[Proposition 3.1]{AssadiD21} for some details).

On the other hand,
a recent result~\cite[Theorem 1.5]{MandePSS24} shows that
$\max_i \set{a_i + b_i}$ can indeed be found in $O(n \log \log n)$
communication.
Since finding the maximum degree (when the edges are partitioned between
Alice and Bob) can easily be reduced to this problem, its communication
complexity is $\Theta(n \log \log n)$.

Next, can we remove the assumption that the edges are partitioned between
Alice and Bob?
\begin{problem}
  What is the communication complexity of $(\Delta + 1)$-coloring when the
  edge sets $E_A$ and $E_B$ are \emph{not} disjoint?
\end{problem}
Typically, when we drop this assumption, even trivial decision problems
(is $G$ a clique, or a clique minus an edge) balloon to $\Omega(n^2)$
communication complexity.
However, we observed in~\Cref{rem:non-disj} that \Cref{alg:easy} still works
when the edge sets are not disjoint (albeit with foreknowledge of $\Delta$),
so the answer is between $n$ and $n \log \Delta$.

\begin{problem}
  What is the deterministic communication complexity of $(\Delta + 1)$-coloring?
\end{problem}
We saw in \Cref{rem:det} that there is an $O(n \log^2 \Delta)$ communication
deterministic protocol.

\begin{remark}\label{rem:det-better}
  In fact, the communication bound in \Cref{lem:set-int-k1} can be improved
  to $O( \log m )$, via a clever optimization to each step of the binary search
  (see~\cite[Exercise 5.10]{KushilevitzN96} for more details), so the answer
  is between $\Omega(n)$ and $O(n \log \Delta)$.
\end{remark}

Finally, we look at two problems where our algorithmic techniques completely
break down.
\begin{problem}
  What is the communication complexity of $\Delta$-coloring?
\end{problem}
The crucial difference between $\Delta$ and $\paren{\Delta + 1}$-coloring is
that a proper partial $\Delta$-coloring may not necessarily extend to a
complete coloring.
This means that the greedy algorithm that is the backbone of both
\Cref{prop:num-colorings} and our protocols no longer works on an arbitrary
permutation of the vertices.
It is possible to adapt Lov\'{a}sz's proof~\cite{Lovasz75} of
Brooks' Theorem to a $O(n \log n + n \log^2 \Delta)$
deterministic communication protocol, using the protocol
of \Cref{rem:det} as a sub-routine.
Again, via the improvement in \Cref{rem:det-better}, the upper bound can
be tightened to $O(n \log n)$.

\begin{problem}
  Let $\kappa$ denote the degeneracy of the input graph $G$.
  What is the communication complexity of $(\kappa + 1)$-coloring?
\end{problem}
Recall that the degeneracy of a graph is the minimum over all permutations
of the maximum left-degree of a vertex.
The classical algorithm to find a $(\kappa + 1)$-coloring is to greedily
color the vertices ordered by a permutation achieving minimum left degree
$\kappa$.

A recent result~\cite{AssadiGLMM24} shows that this permutation can be
found deterministically in $O(n \log^3 n)$ bits of communication, and hence
there is an $O(n \log^3 n + n \log^2 \Delta)$ communication protocol for
$(\kappa + 1)$-coloring.

\begin{problem}
  What is the communication complexity of $(\Delta+1)$-coloring if the algorithm uses $r = o(n)$ rounds?
\end{problem}

Our algorithm requires a lot of interaction between players. This might limit its applicability beyond communication complexity. In distributed graph algorithms, the common measure of efficiency is the number of rounds of communication rather than the total communication (although the two are often roughly equal). It is well known that some communication tasks are round-sensitive~\cite{NisanW93,MiltersenNSW98,SaglamT13}.
To the best of our knowledge, the only algorithm using $o(n)$ rounds of communication for $(\Delta+1)$-coloring is the a one round $O(n\log^3 n)$ communication protocol using palette sparsification~\cite{AssadiCK19}.

\subsection*{Acknowledgement}
We thank Sepehr Assadi for several useful discussions and comments on an earlier
version of this paper.
We thank Prantar Ghosh for telling us about \cite[Theorem 1.5]{MandePSS24},
and Manuel Stoeckl for \Cref{rem:det-better}.
We also thank the anonymous reviews of PODC '24 and Distributed Computing for their careful reading
and helpful comments.

\clearpage

\medskip

\printbibliography

\end{document}